\documentclass{article}
\usepackage[T1]{fontenc}
\usepackage[margin=1in]{geometry}
\usepackage{amsmath}
\usepackage{amssymb}
\usepackage{amsthm}
\usepackage[hidelinks]{hyperref}

\newtheorem{theorem}{Theorem}[section]
\newtheorem{lemma}[theorem]{Lemma}
\newtheorem{corollary}[theorem]{Corollary}

\theoremstyle{definition}
\newtheorem{definition}[theorem]{Definition}
\newtheorem{construction}[theorem]{Construction}
\newtheorem*{notation}{Notation}

\theoremstyle{remark}
\newtheorem{remark}[theorem]{Remark}

\DeclareMathOperator{\Dim}{Dim}
\DeclareMathOperator{\strong}{{str}}
\DeclareMathOperator{\mh}{MH}

\newcommand{\N}{\mathbb{N}}
\newcommand{\Z}{\mathbb{Z}}
\newcommand{\Q}{\mathbb{Q}}

\newcommand{\dimFS}[1]{\dim_{\textup{FS}}^{#1}}
\newcommand{\DimFS}[1]{\Dim_{\textup{FS}}^{#1}}
\newcommand{\rhoFS}[1]{\rho_{\textup{FS}}^{#1}}
\newcommand{\RFS}[1]{R_{\textup{FS}}^{#1}}

\title{Multihead Finite-State Compression}

\author{Neil Lutz\\Swarthmore College}

\date{}

\begin{document}

\maketitle

\begin{abstract}
    This paper develops multihead finite-state compression, a generalization of finite-state compression, complementary to the multihead finite-state dimensions of Huang, Li, Lutz, and Lutz (2025). In this model, an infinite sequence of symbols is compressed by a compressor that produces outputs according to finite-state rules, based on the symbols read by a constant number of finite-state read heads moving forward obliviously through the sequence. The main theorem of this work establishes that for every sequence and every positive integer $h$, the infimum of the compression ratios achieved by $h$-head finite-state information-lossless compressors equals the $h$-head finite-state predimension of the sequence. As an immediate corollary, the infimum of these ratios over all $h$ is the multihead finite-state dimension of the sequence.
\end{abstract}

\section{Introduction}

    We introduce multihead finite-state compression, a generalization of finite-state compression, and use it to prove an exact characterization of multihead finite-state dimension.

    Dai, Lathrop, Lutz, and Mayordomo~\cite{FSD} introduced finite-state dimension using finite-state gamblers that implement betting strategies called \emph{gales}. As Lutz had shown that Hausdorff (fractal) dimension~\cite{Haus19} can be characterized using unrestricted gales~\cite{DCC,DISS}, restricting to gales that are implementable by finite-state gamblers is one way to \emph{effectivize} Hausdorff dimension. This particular effectivization has proven broadly useful, particularly for studying Borel normality (e.g.,~\cite{BecHei2013,BoHiVi2005}), and has motivated sustained subsequent research~\cite{EFDB}. The main theorem of~\cite{FSD}, which our main theorem directly generalizes, proves an exact correspondence between finite-state dimension and the infimum of compression ratios achievable by information-lossless finite-state compressors, as defined by Huffman~\cite{Huff59a}.

    In the long tradition of work on multihead finite-state automata (surveyed in~\cite{HKM09}), Huang, Li, Lutz, and Lutz~\cite{MFSD} recently defined multihead finite-state gamblers, which generalize the finite-state gamblers of~\cite{FSD}. For each $h\in\Z^+$, an $h$-head finite-state gambler updates its state based on the current symbol, read by its \emph{leading head}, and on the $h-1$ earlier symbols in the sequence, read by its $h-1$ \emph{trailing heads}. The trailing heads move according to oblivious (i.e., data-independent) rules. For each $h\in\Z^+$,~\cite{MFSD} used these $h$-head finite-state gamblers to define $h$-head finite-state predimension $\dimFS{(h)}(S)$ of a sequence $S$ and defined the multihead finite-state dimension $\dimFS{\mh}(S)$ to be the infimum, over all $h\in\Z^+$, of $\dimFS{(h)}(S)$; we restate these definitions in more detail in Section~\ref{sec:MFSD}. The main result of~\cite{MFSD} is a hierarchy theorem proving that the predictive power of multihead finite-state gamblers increases strictly with the number of heads, analogous to the classic theorem of Yao and Rivest~\cite{YR78} showing that additional heads allow multihead automata to recognize more languages.

    In Section~\ref{sec:MFSC}, we define the complementary notion of multihead finite-state compressors, reusing the head movement infrastructure of~\cite{MFSD} but replacing betting functions with output functions. Whereas a betting function maps each state to a probability distribution on the alphabet to predict the next symbol, an output function maps each (state, symbol) pair to a binary string. For a compressor $C$ and an infinite sequence $S$ over a finite alphabet $\Sigma$, let $C(S[0:n])$ denote the concatenation of $C$'s outputs on the first $n$ symbols of $S$. We consider the asymptotic behavior of the compression ratio $|C(S[0:n])|/n$ as $n$ approaches infinity, defining $\rhoFS{(h)}(S)$ as the infimum, over all $h$-head finite-state compressors $C$, of the limit inferior of this compression ratio.

    Our main theorem states that this quantity is exactly the $h$-head finite-state predimension defined in~\cite{MFSD}:
    \[\dimFS{(h)}(S)=\rhoFS{(h)}(S).\]
    The main theorem of~\cite{FSD} is the $h=1$ case. Defining $\rhoFS{\mh}(S)=\inf_{h\in\Z^+}\rhoFS{(h)}(S)$, it is also immediate from our main theorem that
    \[\dimFS{\mh}(S)=\rhoFS{\mh}(S)\]
    and, by the hierarchy theorem of~\cite{MFSD}, that ``$h+1$ heads are better than $h$'' holds in this compression setting, in the sense that for all $h\in\Z^+$ there is a sequence $S$ satisfying
    \[\rhoFS{(h)}(S)>\rhoFS{(h+1)}(S).\]
    We also define a version of multihead finite-state compression that is \emph{strong} in the sense of Athreya, Hitchcock, Lutz, and Mayordomo~\cite{AHLM} (i.e., approaches the target compression ratio cofinitely often), and we show an exact correspondence to the multihead finite-state strong predimensions defined in~\cite{MFSD}, thereby generalizing the equivalence $\DimFS{}(S)=\RFS{}(S)$ proved in~\cite{AHLM}.

    To prove our main theorem, we define two explicit constructions: an $h$-head finite-state gambler based on an underlying $h$-head finite-state compressor in Section~\ref{sec:c2g}, and vice versa in Section~\ref{sec:g2c}. At a high level, our constructions follow the framework used in~\cite{FSD} for finite-state dimension and compressibility; working in blocks whose size $k$ parameterizes the quality of the approximation, the gambler bets more on a symbol if it is the beginning of many highly compressible suffixes, and the compressor uses the Shannon--Fano--Elias code of a distribution determined by the gambler's simulated bets on suffixes of length $k$. The key difficulty in adapting this framework to the multihead setting is that the next block cannot be simulated by the leading head without knowing which symbols the trailing heads will read during that block. Accordingly, each of our constructions involves maneuvering the trailing heads of the constructed automaton to eventually be $k$ time steps ahead of the trailing heads of the original automaton, while preserving the obliviousness of their movements.

    Our model is related to the deterministic $k$-transducers of Becher, Carton, and Heiber~\cite{BeCaHe2018} used to define and study a finite-state notion of independence~\cite{AlBeCa2019,BeCaHe2018}, and to the multi-head one-way deterministic finite-state transducers (\textsc{mh}-1DFTs) of Raszyk, Basin, and Traytel~\cite{RaBaTr2019}, generalizations of which have been used in runtime verification~\cite{RBKT2019,RaBaTr2020}. In contrast to both of these prior models, our $h$-FSCs apply their output function exactly when the leading head advances and move their heads obliviously. Unlike $k$-transducers, $h$-FSCs do not have multiple tapes. Unlike \textsc{mh}-1DFTs, our $h$-FSCs operate on infinite sequences, without an end-of-string marker, and are used, in this paper, entirely in the context of information-lossless compression.

\section{Multihead Finite-State Compressors}\label{sec:MFSC}

    \begin{definition}\label{def:hfsc}
        For $h\in\Z^+$, an \emph{$h$-head finite-state compressor ($h$-FSC)} is a 6-tuple
        \[C=(T \times Q, \Sigma,\delta, \mu,\nu,(t_0, q_0)),\]
        whose components are defined as follows.
        \begin{itemize}
            \item $T\times Q$ is a non-empty finite \emph{state space}. We refer to the elements of $T$ and $Q$ as \emph{$T$-states} and \emph{$Q$-states}, respectively; the two components are separated to enforce the obliviousness of the trailing heads' movements.
            \item $\Sigma$ is a finite \emph{alphabet} with $|\Sigma|\geq 2$.
            \item $\delta:T\times Q\times\Sigma^h\to T\times Q$ is a \emph{transition function} such that there exist functions $\delta_T:T\to T$ and $\delta_Q:Q\times\Sigma^h\to Q$ with
            \[\delta(t,q,\vec\sigma)=(\delta_T(t),\delta_Q(q,\vec\sigma))\]
            for all $t\in T$, $q\in Q$, and $\vec\sigma\in\Sigma^h$.
            \item $\mu:T\to\{0,1\}^{h-1}$ is a \emph{movement function}.
            \item $\nu:Q\times\Sigma\to\{0,1\}^*$ is an \emph{output function}.
            \item $(t_0,q_0)\in T\times Q$ is an \emph{initial state}.
        \end{itemize}
    \end{definition}
    As in the $h$-head finite-state gamblers of~\cite{MFSD}, at each time step $n$ there is a \emph{position vector} $\pi(n)=(\pi_1(n),\ldots,\pi_{h-1}(n))\in\N^{h-1}$ defined by $\pi(0)=(0,\ldots,0)$ and, for all $n\in\N$,
    \[\pi(n+1)=\pi(n)+\mu(\delta^n_T(t_0)),\]
    where $\delta_T^n$ is $\delta_T$ applied $n$ times. We interpret $(\pi_1(n),\ldots,\pi_{h-1}(n))$ as the positions of $h-1$ \emph{trailing heads} when the \emph{leading head} is at position $n$. Given a sequence $S\in\Sigma^\omega$, the vector of symbols at these positions (including the leading head's) is
    \[\vec\sigma_n=(S[\pi_1(n)],\ldots,S[\pi_{h-1}(n)],S[n]),\]
    and the compressor's state sequence $(t_0,q_0),(t_1,q_1),(t_2,q_2),\ldots$ evolves according to
    \begin{equation}\label{eq:trajectory}
        (t_{n+1},q_{n+1})=(\delta_T(t_n),\delta_Q(q_n,\vec\sigma_n))
    \end{equation}
    for all $n\in\N$.

    As observed in~\cite{MFSD}, these movement rules imply that for each trailing head $i$ there is a constant \emph{speed} $\alpha_i\in[0,1]$ such that
    \begin{equation}\label{eq:speed}
        \alpha_i n-|T|\leq \pi_i(n)\leq \alpha_i n+|T|
    \end{equation}
    holds for all $n\in\N$.

    \begin{definition}
        The \emph{output} of an $h$-head finite-state compressor
        \[C=(T \times Q, \Sigma,\delta, \mu,\nu,(t_0, q_0))\]
        on input $w\in\Sigma^*$ is given by
        \[C(w)=\nu(q_0,w[0])\nu(q_1,w[1])\ldots\nu(q_{|w|-1},w[|w|-1])\in\{0,1\}^*.\]
        The $h$-FSC $C$ is \emph{information-lossless} if the function from $\Sigma^*$ to $\{0,1\}^*\times T\times Q$ that maps each string $w$ to $(C(w),t_{|w|},q_{|w|})$ is injective.
    \end{definition}

    \begin{notation}
        For the remainder of the paper, we assume a fixed, finite, non-unary alphabet $\Sigma$. We use Python-style list slice notation: given a sequence $S[0]S[1]S[2]\ldots$, for all $a,b\in\N$ we let $S[a:b]$ denote the string $S[a]S[a+1]\ldots S[b-1]$. When $b\leq a$, $S[a:b]=\lambda$, the empty string. Slices of finite strings and vectors are handled analogously. All logarithms are base-2.
    \end{notation}

    \begin{definition}\label{def:compressionratios}
        For each $h\in\Z^+$, let $\mathcal{C}^{(h)}$ be the collection of all information-lossless $h$-head finite-state compressors.
        \begin{itemize}
            \item For each $h\in\Z^+$ and $S\in\Sigma^\omega$, we define
            \[\rhoFS{(h)}(S)=\inf_{C\in\mathcal{C}^{(h)}}\liminf_{n\to\infty}\frac{|C(S[0:n])|}{n\log|\Sigma|}\]
            and
            \[\RFS{(h)}(S)=\inf_{C\in\mathcal{C}^{(h)}}\limsup_{n\to\infty}\frac{|C(S[0:n])|}{n\log|\Sigma|}.\]
            \item For each $S\in\Sigma^\omega$, we define
            \[\rhoFS{\mh}(S)=\inf_{h\in\Z^+}\rhoFS{(h)}(S)\]
            and
            \[\RFS{\mh}(S)=\inf_{h\in\Z^+}\RFS{(h)}(S).\]
        \end{itemize}
    \end{definition}
    Note that since the finite-state compressors defined in~\cite{FSD} are exactly our 1-head finite-state compressors, we have $\rhoFS{(1)}(S)=\rhoFS{}(S)$ as defined in~\cite{FSD} and $\RFS{(1)}(S)=\RFS{}(S)$ as defined in~\cite{AHLM}.

    \section{Multihead Finite-State Gamblers and Dimensions}\label{sec:MFSD}

    We now recall the definitions of multihead finite-state gamblers, predimensions, and dimensions introduced by Huang, Li, Lutz, and Lutz~\cite{MFSD}. The term predimension was used in~\cite{MFSD} because when applied to sets, the authors showed these quantities lack an essential property of fractal dimension notions: stability under unions. They also showed that multihead finite-state dimension, by contrast, is stable under unions.

    \begin{definition}
        For $h \in \Z^+$, an \emph{$h$-head finite-state gambler ($h$-FSG)} is a 7-tuple
        \[G = (T \times Q, \Sigma, \delta, \mu, \beta, (t_0, q_0), c_0),\]
        whose components are defined as follows.
        \begin{itemize}
            \item $T$, $Q$, $\Sigma$, $\delta$, $\mu$, and $(t_0,q_0)$ are as in Definition~\ref{def:hfsc}.
            \item $\beta:Q\to\Delta_\Q(\Sigma)$ is a \emph{betting function}; $\Delta_\Q(\Sigma)$ is the collection of all rational-valued discrete probability distributions on $\Sigma$.
            \item $c_0$ is the gambler's \emph{initial capital}.
        \end{itemize}
    \end{definition}  
    The head movement and state transition rules are the same as those of $h$-head finite-state compressors, with the gambler's state sequence
    \[(t_0,q_0),(t_1,q_1),(t_2,q_2),\ldots\]
    on an input sequence $S\in\Sigma^\omega$ also evolving according to~\eqref{eq:trajectory}.

    \begin{definition}
        For $h\in\Z^+$ and
        \[G=(T\times Q,\Sigma,\delta,\mu,\beta,(t_0,q_0),c_0)\]
        an $h$-FSG, the \emph{martingale of $G$} is defined recursively by $d_G(\lambda)=c_0$ (where $\lambda$ denotes the empty string) and, for all $w\in\Sigma^*$ and $a\in\Sigma$
        \[d_G(wa)=|\Sigma|d_G(w)\beta(q_{|w|})(a),\]
        where $q_{|w|}$ is defined according to~\eqref{eq:trajectory}.
        For $s\in(0,\infty)$, the \emph{$s$-gale of $G$} is defined by
        \[d^{(s)}_G(w)=|\Sigma|^{(s-1)|w|}d_G(w)\]
        for all $w\in\Sigma^*$.
    \end{definition}

    \begin{definition}
        For $s\in(0,\infty)$, an \emph{$s$-gale} over an alphabet $\Sigma$ is a function $d:\Sigma^*\to[0,\infty)$ satisfying, for all $w\in\Sigma^*$,
        \[d(w)=|\Sigma|^{-s}\sum_{a\in\Sigma} d(wa).\]
        A \emph{martingale} is a 1-gale. The \emph{success set} of an $s$-gale $d$ is
        \[\mathcal{S}^\infty[d]=\left\{S\in\Sigma^\omega\;\middle|\; \limsup_{n\to\infty}d(S[0:n])=\infty\right\},\]
        and the \emph{strong success set} of $d$ is
        \[\mathcal{S}_{\strong}^\infty[d]=\left\{S\in\Sigma^\omega\;\middle|\; \liminf_{n\to\infty}d(S[0:n])=\infty\right\}.\]
    \end{definition}

    \begin{definition}\label{def:mfsd}
        For $h\in\Z^+$ and $S\in\Sigma^\omega$, the \emph{$h$-head finite-state predimension} of $S$ is
        \[\dimFS{(h)}(S)=\inf\left\{s\;\middle|\; \exists\ h\text{-FSG }G\text{ such that }S\in\mathcal{S}^\infty\big[d^{(s)}_G\big]\right\}\]
        and the \emph{$h$-head finite-state strong predimension} of $S$ is
        \[\DimFS{(h)}(S)=\inf\left\{s\;\middle|\; \exists\ h\text{-FSG }G\text{ such that }S\in\mathcal{S}_{\strong}^\infty\big[d^{(s)}_G\big]\right\}.\]
        For $S\in\Sigma^\omega$, the \emph{multihead finite-state dimension} of $S$ is
        \[\dimFS{\mh}(S)=\inf_{h\in\Z^+}\dimFS{(h)}(S),\]
        and the \emph{multihead finite-state strong dimension} of $S$ is
        \[\DimFS{\mh}(S)=\inf_{h\in\Z^+}\DimFS{(h)}(S).\]
    \end{definition}

\section{From Compressor to Gambler}\label{sec:c2g}

    \begin{construction}\label{const:gambler}
        Given any information-lossless $h$-FSC
        \[C=(T \times Q, \Sigma,\delta, \mu,\nu,(t_0, q_0)),\]
        and $k\in\Z^+$, we define the $h$-FSG
        \[G_k=(T'\times Q',\Sigma,\delta',\mu',\beta,(t_0',q_0'),1)\]
        as follows.

        \paragraph{Head movements.}
        Assume that $\alpha_1,\ldots,\alpha_{h-1}\in[0,1)$, as any trailing head $i$ with $\alpha_i=1$ can be simulated directly by the leading head. Hence, we define
        \[n_0=\max_{i<h}\min\{n\in\N:n\geq\pi_i(n+k)\},\]
        Informally, $n_0$ is the number of time steps it takes for the leading head to get at least $k$ time steps ahead of all trailing heads. Note that because trailing heads do not move in every time step, staying $k$ time steps ahead is not the same as staying $k$ positions ahead.
        \begin{remark}\label{rem:n0}
            If $\alpha<1$ is the maximum speed of a trailing head in $C$, then for all
            \[n\geq \frac{\alpha k+|T|}{1-\alpha},\]
            for each trailing head $i$, we have
            \begin{align*}
                \pi_i(n+k)&\leq \alpha\cdot(n+k)+|T|\\
                &\leq n,
            \end{align*}
            by~\eqref{eq:speed}. Hence, holding the compressor $C$ constant, we have $n_0=O(k)$.
        \end{remark}

        Let $T'=T\times\{0,\ldots,n_0\}$, and define $\delta'_{T'}:T'\to T'$ by
        \[
            \delta'_{T'}(t,n)=
            \begin{cases}
                (\delta_T(t),n+1)&\text{if }n<n_0\\
                (\delta_T(t),n_0)&\text{if }n=n_0
            \end{cases}
        \]
        and $\mu':T'\to\{0,1\}^{h-1}$ by $\mu'(t,n)=(\mu'_1(t,n),\ldots,\mu'_{h-1}(t,n))$, where
        \[\mu'_i(t,n)=
            \begin{cases}
                1&\text{if }n<n_0\text{ and }n<\pi_i(n+k)\\
                \mu_i(\delta_T^k(t))&\text{otherwise}.
            \end{cases}
        \]
        For each $n\in\N$, let $\pi'(n)=(\pi'_1(n),\ldots,\pi'_{h-1}(n))$ be the position vector induced by $\mu'$, observing that $\pi'_i(n)=\pi_i(n+k)$ holds for all $n\geq n_0$. That is, each trailing head in $G_k$ will eventually stay exactly $k$ time steps ahead of the corresponding trailing head in $C$. Define $t_0'=(t_0,0)$.

    \paragraph{$Q'$-state transitions.}
        Let
        \[\ell=k\left\lceil\frac{n_0+k}{k}\right\rceil,\]
        noting that $\ell=O(k)$ by Remark~\ref{rem:n0}, and let
        \[Q'=Q\times(\Sigma^{h-1})^k\times \{0,\ldots,k-1\}\times\Sigma^{\leq \ell}.\]
        Our goal is for each entry $\vec{u}\in(\Sigma^{h-1})^k$ to eventually be a record of the vectors of symbols observed by the $h-1$ trailing heads of the gambler $G_k$ in each of the $k$ previous time steps, which, after passing the threshold $\ell$, will equal the vectors of symbols observed by the trailing heads of the compressor $C$ in each of the $k$ \emph{next} time steps.

        To this end, define $\delta'_{Q'}:Q'\times\Sigma^h\to Q'$ by
        \[
            \delta'_{Q'}\big((q,\vec{u},j,w),\vec\sigma\big)=
            \begin{cases}
                \big(q_{|w|+1}(w\vec\sigma[h]),\vec{u}[1:k]\vec\sigma[1:h],(j+1)\bmod k,w\vec\sigma[h]\big)&\text{if }|w|<\ell\\
                \big(\delta_Q(q,\vec{u}[0]\vec\sigma[h]),\vec{u}[1:k]\vec\sigma[1:h],(j+1)\bmod k,w\big)&\text{if }|w|=\ell,
            \end{cases}
        \]
        where $q_{|w|+1}(w\vec\sigma[h])\in Q$ is the $Q$-state reached by $C$ after reading the first $|w|+1$ symbols of any sequence with $w\vec\sigma[h]$ as a prefix.
        
        Informally, $\delta'_{Q'}$ hard-codes the behavior of $C$ on short (length $\leq\ell$) prefixes in order to give $G_k$ time to move its trailing heads into position sufficiently ahead of $C$'s trailing heads, build up a length-$k$ record of trailing head observation vectors, and reach the end of a length-$k$ block. After this initial setup phase, $\delta'_{Q'}$ does the following in each time step.
        \begin{enumerate}
            \item Combine the oldest stored vector of observations by the trailing heads, $\vec{u}[0]\in\Sigma^{h-1}$, with the new symbol observed by the leading head, $\vec\sigma[h]\in\Sigma$, to simulate the corresponding $Q$-state transition from $C$ as if $\vec{u}[0]\vec\sigma[h]\in\Sigma^h$ were $C$'s new vector of observations from all heads.
            \item Delete the oldest stored vector of observations by the trailing heads, $\vec{u}[0]$, from the beginning of record of observations by the trailing heads:
            \[\vec{u}[0:k]\mapsto\vec{u}[1:k].\]
            \item Store the new vector $\vec\sigma[1:h]$ of observations by the trailing heads at the end of the record of observations, restoring the record's length to $k$:
            \[\vec{u}[1:k]\mapsto \vec{u}[1:k]\vec\sigma[1:h].\]
            \item Increment the counter $j$ (modulo $k$) to indicate the leading head's position inside a length-$k$ block.
            \item Leave $w$ unchanged; after the setup phase, $w=S[0:\ell]$ is constant.
        \end{enumerate}
        
        Define $q_0'=(q_0,\vec{u}_0,0,\lambda)$, where $\vec{u}_0\in(\Sigma^{h-1})^k$ is arbitrary; the record $\vec{u}_0$ will be completely overwritten before being used, so this initial value is irrelevant.

        \paragraph{Betting function.}
        The betting function $\beta:Q'\to\Delta_\Q(\Sigma)$ places only uniform bets during the setup phase; for all $(q,\vec{u},j,w)\in Q'$ such that $|w|<\ell$ and all $a\in\Sigma$,
        \[\beta(q,\vec{u},j,w)(a)=\frac{1}{|\Sigma|}.\]
        To define the bets $\beta$ places after that setup phase, when $|w|=\ell$, we first introduce some additional notation; we define the following partial functions to describe the way $G_k$ simulates the behavior of $C$ across multiple time steps.
        \begin{itemize}
            \item $\hat\delta:Q\times(\Sigma^{h-1})^*\times\Sigma^{*}\rightharpoonup Q$ is defined, for all $q\in Q$ and $\vec{u}\in(\Sigma^{h-1})^*$, recursively by $\hat\delta(q,\vec{u},\lambda)=q$
            and, for all $\vec\tau\in\Sigma^{h-1}$, $w\in\Sigma^{|\vec{u}|}$, and $a\in \Sigma$,
            \[\hat\delta(q,\vec{u}\vec\tau,wa)=\delta_Q(\hat\delta(q,\vec{u},w),\vec\tau a).\]
            \item Extend $\hat\delta:Q\times(\Sigma^{h-1})^*\times\Sigma^*\rightharpoonup Q$ by defining, for all $q\in Q$, $\vec{u}\in(\Sigma^{h-1})^*$, and $w\in\Sigma^{<|\vec{u}|}$,
            \[\hat\delta(q,\vec{u},w)=\hat\delta(q,\vec{u}[0:|w|],w).\]
            That is, when $w$ is shorter than the record $\vec{u}$ of observation vectors, we use the $|w|$ oldest observation vectors in $\vec{u}$ to compute $\hat\delta(q,\vec{u},w)$. Note that this is consistent with our definition of $\hat\delta(q,\vec{u},\lambda)$.
            \item $\hat\nu:Q\times (\Sigma^{h-1})^*\times\Sigma^*\rightharpoonup\{0,1\}^*$
            is defined, for all $q\in Q$ and $\vec{u}\in(\Sigma^{h-1})^*$, recursively by $\hat\nu(q,\vec{u},\lambda)=\lambda$
            and, for all $w\in\Sigma^{\leq |\vec{u}|}$ and $a\in \Sigma$,
            \[\hat\nu(q,\vec{u},wa)=\hat\nu(q,\vec{u},w)\nu(\hat\delta(q,\vec{u},w),a).\]
            \item $\gamma:Q\times(\Sigma^{h-1})^*\times 2^{\Sigma^*}\rightharpoonup [0,\infty)\cap\Q$ is defined, for all $q\in Q$, $\vec{u}\in(\Sigma^{h-1})^*$, and $A\subseteq\Sigma^{\leq|\vec{u}|+1}$, by
            \[\gamma(q,\vec{u},A)=\sum_{x\in A}2^{-|\hat\nu(q,\vec{u},x)|}.\]
        \end{itemize}
        Finally, for all $(q,\vec{u},j,w)\in Q'$ such that $|w|=\ell$, we define
        \[\beta(q,\vec{u},j,w)(a)=\frac{\gamma(q,\vec{u},a\Sigma^{k-j-1})}{\gamma(q,\vec{u},\Sigma^{k-j})}.\]
    \end{construction}

        To prove Lemma~\ref{lem:gamblergeneral}, we first prove the following three lemmas.
    \begin{lemma}\label{lem:gamblersuffix}
        In Construction~\ref{const:gambler}, for all $S\in\Sigma^\omega$ and all $j,k,m\in\N$ with $m\geq\ell/k$ and $j\leq k$,
        \begin{equation}\label{eq:gamblersuffix}
            \frac{d_{G_k}(S[0:mk+j])}{d_{G_k}(S[0:mk])}=|\Sigma|^j2^{-\sum\limits_{i=mk}^{mk+j-1}|\nu(q_i,S[i])|}\frac{\gamma(q_{mk+j},\vec{u}_{mk+j},\Sigma^{k-j})}{\gamma(q_{mk},\vec{u}_{mk},\Sigma^k)}.
        \end{equation}
    \end{lemma}
    \begin{proof}
        Let $S\in\Sigma^\omega$ and $k,m\in\N$. We prove by induction on $j$ that~\eqref{eq:gamblersuffix} holds for all $j\leq k$. The $j=0$ case holds trivially. Fix any $j<k$ and suppose~\eqref{eq:gamblersuffix} holds. Then, letting $n=mk+j$,
        \begin{align*}
            \frac{d_{G_k}(S[0:n+1])}{d_{G_k}(S[0:mk])}&=|\Sigma|\beta(q_n,\vec{u}_n,j,S[0:\ell])(S[n])\frac{d_{G_k}(S[0:n])}{d_{G_k}(S[0:mk])}\\
            &=|\Sigma|\frac{\gamma(q_n,\vec{u}_n,S[n]\Sigma^{k-j-1})}{\gamma(q_n,\vec{u}_{n},\Sigma^{k-j})}\frac{d_{G_k}(S[0:n])}{d_{G_k}(S[0:mk])}\\
            &=|\Sigma|^{j+1}2^{-\sum\limits_{i=mk}^{n-1}|\nu(q_i,S[i])|}\frac{\gamma(q_n,\vec{u}_n,S[n]\Sigma^{k-j-1})}{\gamma(q_{mk},\vec{u}_{mk},\Sigma^k)},
        \end{align*}
        by the induction hypothesis. As
        \begin{align*}
            \gamma(q_n,\vec{u}_n,S[n]\Sigma^{k-j-1})
            &=\sum_{x\in S[n]\Sigma^{k-j-1}}2^{-|\hat\nu(q_n,\vec{u}_n,x)|}\\
            &=\sum_{x\in \Sigma^{k-j-1}}2^{-|\hat\nu(q_n,\vec{u}_n,S[n]x)|}\\
            &=\sum_{x\in \Sigma^{k-j-1}}2^{-|\nu(q_n,S[n])\hat\nu(q_{n+1},\vec{u}_{n+1},x)|}\\
            &=2^{-|\nu(q_n,S[n])|}\gamma(q_{n+1},\vec{u}_{n+1},\Sigma^{k-j-1}),
        \end{align*}
        this shows
        \[\frac{d_{G_k}(S[0:n+1])}{d_{G_k}(S[0:mk])}=|\Sigma|^{j+1}2^{-\sum\limits_{i=mk}^{n}|\nu(q_i,S[i])|}\frac{\gamma(q_{n+1},\vec{u}_{n+1},\Sigma^{k-j-1})}{\gamma(q_{mk},\vec{u}_{mk},\Sigma^k)}.\]
        That is,~\eqref{eq:gamblersuffix} holds with $j+1$ replacing $j$. By induction, we conclude that~\eqref{eq:gamblersuffix} holds for all $j\leq k$.
        \qed
    \end{proof}

    \begin{lemma}\label{lem:gamblerblocks}
        In Construction~\ref{const:gambler}, for all $S\in\Sigma^\omega$ and all $k,m\in\N$,
        \begin{equation}\label{eq:gamblerblocks}
            d_{G_k}(S[0:mk])\geq \frac{|\Sigma|^{mk-\ell}2^{-|C(S[0:mk])|}}{\prod_{i=\ell/k}^{m-1}\gamma(q_{ik},\vec{u}_{ik},\Sigma^k)}
        \end{equation}
    \end{lemma}
    \begin{proof}
        Let $S\in\Sigma^\omega$ and $k\in\N$. We prove by induction on $m$ that~\eqref{eq:gamblerblocks} holds for all $m\in\N$. The martingale value is 1 on all strings of length $\leq\ell$, so~\eqref{eq:gamblerblocks} holds trivially when $m\leq\ell/k$. Fix any $m\geq \ell/k$ and suppose~\eqref{eq:gamblerblocks} holds. Then applying Lemma~\ref{lem:gamblersuffix} with $j=k$ and letting $n=(m+1)k$ gives
        \begin{align*}
            d_{G_k}(S[0:n])
            &\geq 2^{-\sum\limits_{i=mk}^{n-1}|\nu(q_i,S[i])|}\frac{\gamma(q_{n},\vec{u}_{n},\{\lambda\})}{\gamma(q_{mk},\vec{u}_{mk},\Sigma^k)}\frac{|\Sigma|^{n-\ell}2^{-|C(S[0:mk])|}}{\prod_{i=\ell/k}^{m-1}\gamma(q_{ik},\vec{u}_{ik},\Sigma^k)}\\
            &=\frac{|\Sigma|^{n-\ell}2^{-|C(S[0:n])|}}{\prod_{i=\ell/k}^{m}\gamma(q_{ik},\vec{u}_{ik},\Sigma^k)}.
        \end{align*}
        That is,~\eqref{eq:gamblerblocks} holds with $m+1$ replacing $m$. By induction, we conclude that~\eqref{eq:gamblerblocks} holds for all $m\in\N$.
        \qed
    \end{proof}
    \begin{lemma}\label{lem:gamblerdivisible}
        In Construction~\ref{const:gambler}, if $C$ is information-lossless, then for all $S\in\Sigma^\omega$ and all $k,m\in\N$,
        \[d_{G_k}(S[0:mk])\geq \frac{|\Sigma|^{mk-\ell}2^{-|C(S[0:mk])|}}{((kM+1)|T||Q|)^m},\]
        where $M=\max\{|\nu(q,a)|\mid q\in Q,\,a\in\Sigma\}$.
    \end{lemma}
    \begin{proof}
        We closely follow the proof of Lemma 7.6 in~\cite{FSD}, which itself follows the proof of Lemma 2 in~\cite{LZ78}. Assume $C$ is information-lossless, and let $S\in\Sigma^\omega$ and $k,m\in\N$. For all $q\in Q$, $\vec{u}\in(\Sigma^{h-1})^k$, and $z\in\{0,1\}^*$,
        \[\left|\left\{w\in\Sigma^k:\hat\nu(q,\vec{u},w)=z\right\}\right|\leq |T\times Q|,\]
        as $C$ is information-lossless. Hence, for each $i\leq m$,
        \begin{align*}
            \gamma(q_{ik},\vec{u}_{ik},\Sigma^k)&=\sum_{w\in\Sigma^k}2^{-|\hat\nu(q_{ik},\vec{u}_{ik},w)|}\\
            &=\sum_{z\in\{0,1\}^{\leq kM}}\left|\left\{w\in\Sigma^k\mid\hat\nu(q_{ik},\vec{u}_{ik},w)=z\right\}\right|2^{-|z|}\\
            &\leq|T||Q|\sum_{z\in\{0,1\}^{\leq kM}}2^{-|z|}\\
            &=|T||Q|\sum_{L=0}^{kM}\sum_{z\in\{0,1\}^L}2^{-L}\\
            &=(kM+1)|T||Q|.
        \end{align*}
        Therefore, by Lemma~\ref{lem:gamblerblocks},
        \begin{align*}
            d_{G_k}(S[0:mk])&\geq \frac{|\Sigma|^{mk-\ell}2^{-|C(S[0:mk])|}}{((kM+1)|T||Q|)^{m-\ell/k}}\\
            &\geq \frac{|\Sigma|^{mk-\ell}2^{-|C(S[0:mk])|}}{((kM+1)|T||Q|)^m}.
        \end{align*}
        \qed
    \end{proof}

    \begin{lemma}\label{lem:gamblergeneral}
        In Construction~\ref{const:gambler}, if $C$ is information-lossless, then for all $S\in\Sigma^\omega$ and all $k,n\in\N$,
        \[d_{G_k}(S[0:n])\geq 2^{-|C(S[0:n])|}\frac{|\Sigma|^{n-\ell}2^{-kM}}{((kM+1)|T||Q|)^{\lceil n/k\rceil}},\]
        where $M=\max\{|\nu(q,a)|\mid q\in Q,\,a\in\Sigma\}$.
    \end{lemma}

    \begin{proof}
        Assume $C$ is information-lossless, and let $S\in\Sigma^\omega$ and $n\in\N$. Let
        \[j=(k - (n\bmod k))\bmod k,\]
        $m=\frac{n+j}{k}\in\N$, and $z\in\Sigma^j$ minimizing $d_{G_k}(S[0:n]z)$. Then $S[0:n]z\in\Sigma^{mk}$, so Lemma~\ref{lem:gamblerdivisible} applies, and we have
        \begin{align*}
            d_{G_k}(S[0:n])&\geq |\Sigma|^{-j}d_{G_k}(S[0:n]z)\\
            &\geq \frac{|\Sigma|^{mk-\ell-j}2^{-|C(S[0:n])|}2^{-jM}}{((kM+1)|T||Q|)^m}\\
            &\geq 2^{-|C(S[0:n])|}\frac{|\Sigma|^{n-\ell}2^{-kM}}{((kM+1)|T||Q|)^{\lceil n/k\rceil}}.
        \end{align*}
    \end{proof}

\section{From Gambler to Compressor}\label{sec:g2c}
    In this section, we show how to construct a compressor whose performance will approximate that of a given gambler. To avoid dividing by zero, it will be useful for the gambler to always make positive bets on each symbol.
    \begin{definition}
        For $h\in\Z^+$, an $h$-FSG
        \[G=(T\times Q,\Sigma,\delta,\mu,\beta,(t_0,q_0),c_0)\]
        is \emph{non-vanishing} if, for all $q\in Q$ and $a\in\Sigma$, $\beta(q)(a)>0$.
    \end{definition}
    This constraint makes little difference as every $h$-FSG can be approximated by a non-vanishing $h$-FSG, in the following sense.
    \begin{lemma}\label{lem:nv}
        For every $h$-FSG
        \[G=(T\times Q,\Sigma,\delta,\mu,\beta,(t_0,q_0),c_0)\]
        and all $\varepsilon>0$, there is a function $\beta':Q\to\Delta_\Q(\Sigma)$ such that
        \[G'=(T\times Q,\Sigma,\delta,\mu,\beta',(t_0,q_0),c_0)\]
        is a non-vanishing $h$-FSG and, for all $S\in\Sigma^\omega$ and $n\in\N$,
        \[d_{G'}(S[0:n])\geq 2^{-\varepsilon n}d_G(S[0:n]).\]
    \end{lemma}
    \begin{proof}
        This is essentially Lemma 3.11 of~\cite{FSD}, generalized to finite (non-unary) alphabets and stated for $h$-head finite-state gamblers. As only the betting function is involved in the transformation $G\mapsto G'$, the presence of additional heads is irrelevant to the argument. 
    \end{proof}

    \begin{construction}\label{const:compressor}
        Let $k\geq 1$, and let
        \[G=(T\times Q,\Sigma,\delta,\mu,\beta,(t_0,q_0),c_0)\]
        be a non-vanishing $h$-head finite-state gambler. We define the $h$-head finite-state compressor
        \[C_k=(T' \times Q',\Sigma,\delta',\mu',\nu,(t_0',q_0'))\]
        as follows.
        
        \paragraph{Head movements.} The components $T'$, $\mu'$, and $t_0'$ are identical to those defined in Construction~\ref{const:gambler}, but now based on the components of the gambler $G$ rather than the compressor $C$. In particular, the trailing heads of $C_k$ will eventually be $k$ time steps ahead of those of $G$, just as the trailing heads of $G_k$ were eventually $k$ time steps ahead of those of $C$ in Construction~\ref{const:gambler}. Let $\ell$ also be as defined in Construction~\ref{const:gambler}.

        \paragraph{$Q'$-state transitions.} The components $Q'$, $\delta'_{Q'}$, and $q'_0$ are also very similar to those defined in Construction~\ref{const:gambler}, with the same intuition of a setup phase and length-$k$ records of trailing head observation vectors. The only difference is that a $Q'$-state will now, within each length-$k$ block, store the string read so far within that block rather than only the length $j$ of that string. Formally,
        \[Q'=Q\times (\Sigma^{h-1})^k\times \Sigma^{\leq k}\times\Sigma^{\leq \ell},\]
        and $\delta'_{Q'}:Q'\times\Sigma^h\to Q'$ is defined by
        \begin{align*}&\delta'_{Q'}\big((q,\vec{u},z,w),\vec\sigma\big)=\\
                &\quad
            \begin{cases}
                    \big(q_{|w|+1}(w\vec\sigma[h]),\vec{u}[1:k]\vec\sigma[1:h],z\vec\sigma[h],w\vec\sigma[h]\big)&\text{if }|w|<\ell\text{ and }|z|<k-1\\
                    \big(q_{|w|+1}(w\vec\sigma[h]),\vec{u}[1:k]\vec\sigma[1:h],\lambda,w\vec\sigma[h]\big)&\text{if }|w|<\ell\text{ and }|z|=k-1\\
                    \big(\delta_Q(q,\vec{u}[0]\vec\sigma[h]),\vec{u}[1:k]\vec\sigma[1:h],z\vec\sigma[h],w\big)&\text{if }|w|=\ell\text{ and }|z|<k-1\\
                    \big(\delta_Q(q,\vec{u}[0]\vec\sigma[h]),\vec{u}[1:k]\vec\sigma[1:h],\lambda,w\big)&\text{if }|w|=\ell\text{ and }|z|=k-1,
            \end{cases}
        \end{align*}
        where $q_{|w|+1}(w\vec\sigma[h])\in Q$ is the $Q$-state reached by $G$ after reading the first $|w|+1$ symbols of any sequence with $w\vec\sigma[h]$ as a prefix.

        Define the initial $Q'$-state of $C_k$ to be $q_0'=(q_0,\vec{u}_0,\lambda,\lambda)$.

        \paragraph{Output function.} During the setup phase, the compressor will not attempt to compress; it will output, for each symbol $a$, the binary string $\Phi(a)$, where $\Phi:\Sigma\to\{0,1\}^{\lceil\log|\Sigma|\rceil}$ is an arbitrary injective function.

        For compression after the setup phase, define the function
        \[\tilde\beta:Q\times(\Sigma^{h-1})^k\times\Sigma^{\leq k}\to[0,1]\]
        recursively by, for all $q\in Q$ and $\vec{u}\in(\Sigma^{h-1})^k$,
        \[\tilde\beta(q,\vec{u},\lambda)=1,\]
        and for all $w\in\Sigma^{<|\vec{u}|}$ and $a\in\Sigma$,
        \[\tilde\beta(q,\vec{u},wa)=\tilde\beta(q,\vec{u},w)\beta(\hat{\delta}(q,\vec{u},w))(a),\]
        where $\hat{\delta}$ is the same function defined in Construction~\ref{const:gambler}, now using the $Q$-state transition function of the gambler $G$ as the underlying $\delta_Q$.

        Then define, for each $q\in Q$ and $\vec{u}\in(\Sigma^{h-1})^k$, the probability measure $\hat\beta_{q,\vec{u}}:\Sigma^k\to[0,1]$ by
        \[\hat\beta_{q,\vec{u}}(w)=\tilde\beta(q,\vec{u},w).\]
        Note that for all $n$, if $\vec{u}=(\vec\sigma_n,\ldots,\vec\sigma_{n+k-1})$, $q=q_n$, and $w=S[n:n+k]$, then $C_k$ will accurately simulate the gambling of $G$ on that interval, in the sense that
        \[\hat\beta_{q,\vec{u}}(w)=\prod_{j=n}^{n+k-1}\beta(q_j)(S[j]).\]
        Let $\Theta_{q,\vec{u}}:\Sigma^k\to\{0,1\}^*$ be the Shannon--Fano--Elias code~(see~\cite{CovTho06}) for the measure $\hat\beta_{q,\vec{u}}$. Then for all $w\in\Sigma^k$,
        \[|\Theta_{q,\vec{u}}(w)|=1+\left\lceil-\log\hat\beta_{q,\vec{u}}(w)\right\rceil.\]
        Using this code, we define the output function $\nu:Q'\times\Sigma\to\{0,1\}^*$ by
        \[
            \nu((q,\vec{u},z,w),a)=
            \begin{cases}
                \Phi(a)&\text{if }|w|<\ell\\
                \lambda&\text{if }|w|=\ell\text{ and }|z|<k-1\\
                \Theta_{q,\vec{u}}(wa)&\text{if }|w|=\ell\text{ and }|z|=k-1.
            \end{cases}
        \]            
    \end{construction}

    \begin{lemma}\label{lem:compressorgeneral}
        In Construction~\ref{const:compressor}, for all $k\in\N$, $C_k$ is information-lossless and, for all $S\in\Sigma^\omega$ and $n\in\N$,
        \[|C_k(S[0:n])|\leq\frac{2n}{k}+(n+\ell)\log|\Sigma|+\ell+\log c_0-\log d_G(S[0:n]).\]
    \end{lemma}
    \begin{proof}
        Let $k,n\in\N$ and $S\in\Sigma^\omega$. The compressor $C_k$ is information-lossless because the range of each $\Theta_{q,\vec{u}}$ is a prefix code. Let $S\in\Sigma^\omega$ and $k,n\in\N$. For each $i\in\N$, let $\vec{u}_i=(\vec\sigma_{ik},\ldots,\vec\sigma_{ik+k-1})$ and $w_i=S[ik:ik+k]$.  Then
        \begin{align*}
            |C_k(S[0:n])|&=\ell\lceil\log|\Sigma|\rceil + \sum_{i=\ell/k}^{\lfloor n/k\rfloor-1}|\Theta_{q_{ik},\vec{u}_i}(w_i)|\\
            &\leq \ell\lceil\log|\Sigma|\rceil + \sum_{i=0}^{\lfloor n/k\rfloor-1}|\Theta_{q_{ik},\vec{u}_i}(w_i)|\\
            &= \ell\lceil\log|\Sigma|\rceil + \sum_{i=0}^{\lfloor n/k\rfloor-1}\left(1+\left\lceil-\log\hat\beta_{q_{ik},\vec{u}_i}(w_i)\right\rceil\right)\\
            &\leq\ell\lceil\log|\Sigma|\rceil + \frac{2n}{k}-\sum_{i=0}^{\lfloor n/k\rfloor-1}\log\hat\beta_{q_{ik},\vec{u}_i}(w_i).
        \end{align*}
        We now bound the sum.
        \begin{align*}
            \sum_{i=0}^{\lfloor n/k\rfloor-1}\log\hat\beta_{q_{ik},\vec{u}_i}(w_i)&=\sum_{i=0}^{\lfloor n/k\rfloor-1}\log\left(\prod_{j=ik}^{ik+k-1}\beta(q_j)(S[j])\right)\\
            &=\sum_{i=0}^{\lfloor n/k\rfloor-1}\sum_{j=ik}^{ik+k-1}\log\big(\beta(q_j)(S[j])\big)\\
            &=\sum_{j=0}^{\lfloor n/k\rfloor k-1}\log\big(\beta(q_j)(S[j])\big)\\
            &\geq\sum_{j=0}^{n-1}\log\big(\beta(q_j)(S[j])\big)\\
            &=\log\left(\prod_{j=0}^{n-1}\beta(q_j)(S[j]) \right)\\
            &=\log\frac{d_G(S[0:n])}{c_0|\Sigma|^n}.
        \end{align*}
        Therefore,
        \[|C_k(S[0:n])|\leq\frac{2n}{k}+(n+\ell)\log|\Sigma|+\ell+\log c_0-\log d_G(S[0:n]).\]
        \qed
    \end{proof}

\section{Main Theorem}\label{sec:main}

    We now prove our main theorem, an equivalence, for each fixed number $h$ of heads, between $h$-head finite-state dimension and $h$-head finite-state compressibility.

    \begin{theorem}\label{thm:main}
        For all $h\in\Z^+$ and $S\in\Sigma^\omega$,
        \[\dimFS{(h)}(S)=\rhoFS{(h)}(S)\]
        and
        \[\DimFS{(h)}(S)=\RFS{(h)}(S).\]
    \end{theorem}

    \begin{proof}
        Let $h\in\Z^+$ and $S\in\Sigma^\omega$.
                
        To prove that $\dimFS{(h)}(S)\leq\rhoFS{(h)}(S)$, let $r=\rhoFS{(h)}(S)$ and $\varepsilon>0$. Then there is some information-lossless $h$-FSC
        \[C=((T\times Q),\Sigma,\delta,\mu,\nu,(t_0,q_0))\]
        such that
        \begin{equation}\label{eq:compressorbound}
            \liminf_{n\to\infty}\frac{|C(S[0:n])|}{n\log|\Sigma|}\leq r+\varepsilon.
        \end{equation}
        Let $M=\max\{|\nu(q,a)|\mid q\in Q,\,a\in\Sigma\}$, let $k\in\N$ be sufficiently large so that
        \[((kM+1)|T||Q|)^{\lceil n/k\rceil}\leq 2^{\varepsilon n}\]
        holds for all $n\geq k$, let $G_k$ be the result of applying Construction~\ref{const:gambler} to $C$ with this $k$, and let $\ell$ be as in this construction. Then by Lemma~\ref{lem:gamblergeneral}, for all $n\geq k$,
        \begin{align*}
            d_{G_k}(S[0:n])&\geq 2^{-|C(S[0:n])|}\frac{|\Sigma|^{n-\ell}2^{-kM}}{((kM+1)|T||Q|)^{\lceil n/k\rceil}}\\
            &\geq 2^{-|C(S[0:n])|-kM-\varepsilon n}|\Sigma|^{n-\ell},
        \end{align*}
        so
        \begin{align*}
            d_{G_k}^{(r+4\varepsilon)}(S[0:n])&\geq 2^{-|C(S[0:n])|-kM-\varepsilon n}|\Sigma|^{n-\ell+(r+4\varepsilon-1)n}\\
            &=2^{-|C(S[0:n])|-kM-\varepsilon n-\ell\log|\Sigma|+(r+4\varepsilon) n\log|\Sigma|}\\
            &\geq 2^{-|C(S[0:n])|-kM-\ell\log|\Sigma|+(r+3\varepsilon) n\log|\Sigma|}.
        \end{align*}
        By~\eqref{eq:compressorbound}, there are infinitely many $n\in\N$ such that
        \[|C(S[0:n])|\leq (r+2\varepsilon)n\log|\Sigma|.\]
        For all such $n$, we have
        \[d_{G_k}^{(r+4\varepsilon)}(S[0:n])\geq 2^{-kM-\ell\log|\Sigma|+\varepsilon n\log|\Sigma|},\]
        which diverges as $n$ approaches $\infty$, so
        \[S\in\mathcal{S}^\infty\big[d^{(r+4\varepsilon)}_{G_k}\big].\]
        Letting $\varepsilon$ approach 0, we have shown that $\dimFS{(h)}(S)\leq\rhoFS{(h)}(S)$.

        For the other direction, let $s=\dimFS{(h)}(S)$ and $\varepsilon>0$. Then there is some $h$-head finite-state gambler $G$ such that
        \[S\in \mathcal{S}^\infty\big[d_G^{(s+\varepsilon)}\big],\]
        so by Lemma~\ref{lem:nv}, there is also a non-vanishing $h$-head finite-state gambler $G'$ such that
        \[S\in \mathcal{S}^\infty\big[d_{G'}^{(s+2\varepsilon)}\big].\]
        Without loss of generality, assume $G'$ has initial capital $c_0=1$.
        
        Let $k=\lceil 2/\varepsilon\rceil$, let $C_k$ be the result of applying Construction~\ref{const:compressor} to $G'$ with this $k$, and let $\ell$ be as in this construction. By Lemma~\ref{lem:compressorgeneral}, $C_k$ is information-lossless, and
        \begin{align*}
            |C_k(S[0:n])|&\leq \frac{2n}{k}+(n+\ell)\log|\Sigma|+\ell-\log d_{G'}(S[0:n])\\
            &= \frac{2n}{k}+(n+\ell)\log|\Sigma|+\ell-\log\left(|\Sigma|^{(1-s-2\varepsilon)n} d_{G'}^{(s+2\varepsilon)}(S[0:n])\right)\\
            &\leq \left(\frac{2}{k}+s+2\varepsilon+\frac{(\log|\Sigma|+1)\ell}{n}\right)n\log|\Sigma|-\log d^{(s+2\varepsilon)}_{G'}(S[0:n])\\
            &\leq \left(s+3\varepsilon+\frac{(\log|\Sigma|+1)\ell}{n}\right)n\log|\Sigma|-\log d^{(s+2\varepsilon)}_{G'}(S[0:n])\\
            &\leq (s+4\varepsilon)n\log|\Sigma|-\log d^{(s+2\varepsilon)}_{G'}(S[0:n]),
        \end{align*}
        for all sufficiently large $n$. Since $S\in \mathcal{S}^\infty\big[d_{G'}^{(s+2\varepsilon)}\big]$, there are infinitely many $n\in\N$ satisfying
        \[d_{G'}^{(s+2\varepsilon)}(S[0:n])\geq 1.\]
        For each such $n$, we have
        \[|C_k(S[0:n])|\leq (s+4\varepsilon)n\log|\Sigma|,\]
        so
        \[\liminf_{n\to\infty}\frac{|C_k(S[0:n])|}{n\log|\Sigma|}\leq (s+4\varepsilon).\]
        Letting $\varepsilon$ approach $0$, we have shown that $\rhoFS{(h)}(S)\leq\dimFS{(h)}(S)$.
        
        By replacing
        \[\rhoFS{(h)},\ \dimFS{(h)},\ \mathcal{S}^\infty,\ \text{infinitely},\ \liminf\]
         with
        \[\RFS{(h)},\ \DimFS{(h)},\ \mathcal{S}^\infty_{\strong},\ \text{cofinitely},\ \limsup,\]
        the same arguments given above prove that $\RFS{(h)}(S)=\DimFS{(h)}(S)$.    
        \qed
    \end{proof}
    Setting $h=1$ in Theorem~\ref{thm:main} yields $\dimFS{}(S)=\rhoFS{}(S)$, the main theorem of~\cite{FSD}, and the strong dimension version, $\DimFS{}(S)=\RFS{}(S)$, proved in~\cite{AHLM}.
    
    The following corollaries of Theorem~\ref{thm:main} are immediate by Definitions~\ref{def:mfsd} and~\ref{def:compressionratios} and by the hierarchy theorem of~\cite{MFSD}, respectively.
    \begin{corollary}
        For all $S\in\Sigma^\omega$,
        \[\dimFS{\mh}(S)=\rhoFS{\mh}(S)\]
        and
        \[\DimFS{\mh}(S)=\RFS{\mh}(S).\]
    \end{corollary}
    \begin{corollary}
        For each $h\in\Z^+$, there is some sequence $S\in\Sigma^\omega$ such that
        \[\rhoFS{(h)}(S)>\rhoFS{(h+1)}(S)\]
        and
        \[\RFS{(h)}(S)>\RFS{(h+1)}(S).\]
    \end{corollary}

\bibliographystyle{plainurl}
\bibliography{mfspc}

\end{document}